\documentclass[11pt]{article}
\usepackage{graphicx}
\usepackage{amsmath,amsthm}
\usepackage{amsfonts}
\usepackage{amssymb}
\usepackage{url}
\usepackage{hyperref}  %This only works for pdflatex, not latex
\usepackage{verbatim}
\usepackage{color}
\definecolor{red}{rgb}{1,0,0}

\newtheorem{thm}{Theorem}[section]
\newtheorem{defn}[thm]{Definition}
\newtheorem{prop}[thm]{Proposition}
\newtheorem{cor}[thm]{Corollary}
\newtheorem{lem}[thm]{Lemma}

\newtheorem{ex}[thm]{Example}

\newtheorem{rem}[thm]{Remark}

\def\mtx#1{\begin{bmatrix} #1 \end{bmatrix}} % matrix
\def\ord#1{| #1 |} % order of G is |G|
\def\lie #1{\langle #1 \rangle_{[\cdot,\cdot]}} % complement

\newcommand{\R}{\mathbb{R}}

\newcommand{\C}{\mathbb{C}}

\newcommand{\Cnn}{\C^{n\times n}}
\newcommand{\Cn}{\C^{n}}
\newcommand{\Rn}{\R^{n}}
\newcommand{\Rnn}{\R^{n\times n}}

\newcommand{\G}{\mathcal{G}}
\newcommand{\LL}{\mathcal{L}}
\newcommand{\KK}{\mathcal{K}}

\newcommand{\Hn}{\mathcal{H}_n}
\newcommand{\sym}{\mathcal{S}}

\newcommand{\W}{\widetilde W}
\newcommand{\rank}{\operatorname{rank}} % not needed in SIAM
\newcommand{\range}{\operatorname{range}}
\newcommand{\Span}{\operatorname{span}}

\newcommand{\Z}{\operatorname{Z}}

\newcommand{\tr}{\operatorname{tr}}

\newcommand{\bx}{{\bf x}}
\newcommand{\by}{{\bf y}}

\newcommand{\bz}{{\bf z}}
\newcommand{\bu}{{\bf u}}
\newcommand{\bb}{{\bf b}}

\newcommand{\be}{{\bf e}}

\newcommand{\bit}{\begin{itemize}}
\newcommand{\eit}{\end{itemize}}
\newcommand{\ben}{\begin{enumerate}}
\newcommand{\een}{\end{enumerate}}
\newcommand{\beq}{\begin{equation}}
\newcommand{\eeq}{\end{equation}}
\newcommand{\bea}{\begin{eqnarray*}}
\newcommand{\eea}{\end{eqnarray*}}
\newcommand{\bean}{\begin{eqnarray}}
\newcommand{\eean}{\end{eqnarray}}
\newcommand{\bpf}{\begin{proof}}
\newcommand{\epf}{\end{proof}\ms}
\newcommand{\ms}{\medskip}

\newcommand{\x}{\times}

\begin{document}

\title{Zero forcing, linear and quantum controllability for systems evolving on networks}

\author{Daniel Burgarth\thanks{Institute of Mathematics and Physics, Aberystwyth University, SY23 3BZ Aberystwyth, United Kingdom; \texttt{daniel@burgarth.de}.}\and  Domenico D'Alessandro\thanks{Department of Mathematics, Iowa State University, Ames, IA 50011, USA; \texttt{dmdaless@gmail.com}.} \and Leslie Hogben\thanks{Department of Mathematics, Iowa State University, Ames, IA 50011, USA, and American Institute of Mathematics, 360 Portage Ave, Palo Alto, CA 94306, USA; \texttt{lhogben@iastate.edu; hogben@aimath.org}.} \and Simone Severini\thanks{Department of Computer Science, and Department of Physics \& Astronomy,
University College London, WC1E 6BT London, United Kingdom; \texttt{simoseve@gmail.com}.} \and
Michael Young\thanks{Department of Mathematics, Iowa State University, Ames, IA 50011, USA; \texttt{myoung@iastate.edu}.}}

\maketitle
%\linenumbers

\abstract{We study the dynamics of systems on networks from a linear algebraic
perspective. The control theoretic concept of {\em controllability}
describes the set of states that can be reached for these systems.
Under appropriate conditions, there is a connection between the
quantum (Lie theoretic) property of controllability and the linear
systems (Kalman) controllability condition. We investigate how the
graph theoretic concept of a zero forcing set impacts the
controllability property. In particular, we prove that if a set of
vertices is a zero forcing set, the associated dynamical
system is controllable. The results open up the possibility of further
exploiting the analogy between networks, linear control systems
theory, and quantum systems Lie algebraic theory. This study is
motivated by several quantum systems currently under study, including
continuous quantum walks modeling transport phenomena. Additionally,
it proposes zero forcing as a new notion in the analysis of complex
networks.}

\section{Introduction}\label{sintro}

%We \emph{control} a built device or a natural system, for example,
%physical or biological, if we have the ability to put the system
%into any of its possible configurations, or, in other words, any of
%the points of its phase space. We control a light switch if we have
%the ability to put it ON and OFF. We control a car if we have the
%ability to drive in any desired direction. Given its general
%breadth, the notion of control is therefore applied to the most
%diverse areas, spanning from systems biology, where it helps to
%understand the dynamics of biological regulations, to quantum
%information processing, where it is instrumental for the development
%of nanotechnologies for information transfer.

This paper deals with several concepts from different fields such as linear algebra,  graph theory and quantum and classical (linear) control theory. In the context of dynamics and control of systems on networks, it establishes a connection between a notion of graph theory (zero forcing) and concepts in control theory (quantum and classical controllability). We review these different concepts before we introduce the technical content of the paper and give  physical motivation for our study.

\subsection{Background}

For a  dynamical system with a {\it control input},  the property of {\it
controllability} describes to what extent one can go from one state
to another with the evolution corresponding to an appropriate choice
of the control. If all the possible state transfers can be obtained
within a natural set (the phase space), then the system is said to
be {\it controllable}.

For several classes of systems, controllability has been described
in detail and controllability tests are known. In particular,  for
{\it linear systems}
\begin{equation}\label{sistema12}
\dot \bx = A \bx +\sum_{j=1}^s \bb_j u_j,
\end{equation}
$A \in \Rnn$, $\bb_j \in \Rn$, $j=1,2,\ldots,s$, where both the
state $\bx \in \Rn $ and the control functions $u_j=u_j(t)$ enter
the right hand side linearly, several equivalent conditions of
controllability are known. The classical Kalman controllability
condition (see, e.g., \cite{Kailath}) says the system
(\ref{sistema12}) is controllable if and only if the $n \times (ns)$
matrix
\[\W(A,B):=[\bb_1, A\bb_1,\dots,A^{n-1}\bb_1,\dots,\bb_s, A\bb_s,\dots,A^{n-1}\bb_s], \]
has full rank $n$, where $B:=[\bb_1 \quad \bb_2 \quad \cdots
\bb_s]$. In this case, for any prescribed state transfer $\bx_0
\rightarrow \bx_1 (\in \Rn)$ %, there exists a control $\bu(t)=[u_1,\ldots,u_s]^T$ and an interval $[0,T]$ such that
and interval $[0,T]$, there exists a control $\bu(t)=[u_1,\ldots,u_s]^T$ such that
the corresponding
solution $\bx(t)$ of (\ref{sistema12})  satisfies  $\bx(0)=\bx_0$ and
$\bx(T)=\bx_1$.
For quantum mechanical systems which are closed (i.e., not
interacting with the environment) and finite dimensional, one
considers the {\it Schr\"odinger equation}
\begin{equation}\label{Scriding}
i \frac{d}{dt} |\psi \rangle = H(\bu) |\psi \rangle,
\end{equation}
where $|\psi \rangle \in \Cn$ is the quantum state and the
Hamiltonian matrix $H=H(\bu)$ is Hermitian and depends on a control
$\bu=\bu(t)$ which in some cases can be assumed to be a switch between
different Hamiltonians. If (\ref{Scriding}) is a system linear in
the state $|\psi\rangle$, the solution of (\ref{Scriding}) is
$|\psi(t) \rangle =X(t)|\psi(0)\rangle$ where $X=X(t)$ is the solution of the
{\it Schr\"odinger matrix equation}
\begin{equation}\label{scromat}
i\dot X= H(\bu) X
\end{equation}
with initial condition equal to the $n \times n$ identity matrix $I_n$.
Since $H=H(\bu)$ is Hermitian for every value of $\bu$ and therefore
$-iH$ is skew-Hermitian, the solution of (\ref{scromat}) is forced
to be unitary at every time $t$. In this context, the system is
called completely controllable if for any unitary matrix $X_f$ in
$SU(n)$\footnote{Following
standard notation,  $SU(n)$ is the special unitary group, i.e.,
the matrix group of $n \times n$ unitary matrices having determinant 1.} there exists a control function $\bu=\bu(t)$ and an interval
$[0,T]$ such that the corresponding solution $X=X(t)$ of
(\ref{scromat}) satisfies  $X(0)=I_n$ and $X(T)=X_f$.

At the beginning of the development of the theory of quantum
control, it was realized  (see e.g., \cite{Tarn}) that system
(\ref{scromat}) has a structure familiar in geometric control theory
\cite{Jurdjevic} and therefore controllability conditions developed
there can be directly applied. In particular, the Lie algebra rank
condition \cite{JS} says that a necessary and sufficient condition for
complete controllability of system (\ref{scromat}) is that the Lie
algebra generated by the matrices $\{i H(\bu) \}$ (as $\bu$ varies in
the set of admissible values for the control) is $su(n)$ or $u(n)$.\footnote{Following
standard notation, $u(n)$ is the Lie algebra of $n\times n$ skew-Hermitian matrices and  $su(n)$ is
the Lie algebra of $n \times n$ skew-Hermitian matrices with zero
  trace.} %, $gl(n,\R)$  denotes the Lie algebra of all real $n \times n$   matrices, and $sl(n)$ denotes the Lie algebra of $n\times n$ real matrices with zero trace. All these Lie algebras are considered as vector spaces over the field of real numbers.}
  This has
given rise to a comprehensive approach to  quantum control based on
the application of techniques of  Lie algebras and Lie group
theory \cite{mybook}.

In recent years there has been considerable interest in the study of
control systems, both classical and quantum, which are naturally
modeled on networks. Often one tries to relate the controllability
of these systems to topological or graph theoretic properties of the
network. For quantum systems, the nodes of the network may represent
energy levels or particles which are interacting with each other.
For these systems, the application of the Lie algebra rank condition
to determine controllability can become cumbersome and subject to
errors when the dimension of the system becomes large. It is
preferable to have criteria  based on graph theoretic properties of
the network not only because they are typically checked more
efficiently but also because they give more insight in the dynamics
of the system. Work in this direction has been done in
\cite{AlbDal}, \cite{graphinfect}, \cite{Turinici}. In this context,
a relevant property of a graph $G$ and a subset $S$ of its vertices
is the capability of this set to `infect' all the vertices of the
graph, as explained in the next paragraph. %In this case the set $S$ is called a {\it zero forcing set}.

%Much of our focus is on graphs or matrices described by a particular graph.
Every graph discussed is simple (no loops or multiple edges), undirected, and has a finite nonempty vertex set.
 Consider a graph $G$ and color each of its vertices  black or white. A
vertex $v$ is said to {\em infect}, or {\em force} a vertex $w$ if $v$ is black,
$w$ is white, $w$ is a neighbor of $v$, and $w$ is the only
white neighbor of $v$. In the case where infection of $w$ has
occurred, we change the color of $w$ to black and continue the iterative
procedure. The set $S$ is called a \emph{zero forcing set} if this procedure,
starting from a graph where only the vertices in $S$ are black,  leads to a graph where
{\it all} vertices are black.  An example of a zero forcing (infection) process is shown in Figure \ref{zfs}, indicated by arrows; the set of black vertices is a zero forcing set.
\begin{figure}[!ht]
\begin{center}\scalebox{.3}{\includegraphics{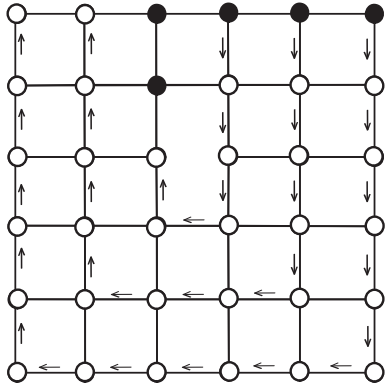}}
\caption{A zero forcing set and the process by which it can infect all vertices.}\label{zfs}
\end{center}
\end{figure}

  For a real symmetric $n\x n$ matrix $A=[a_{kj}]$, the {\em graph} of $A$, denoted $\G(A)$, is the
graph with vertices
$\{1,\dots,n \}$  and edges $\{ kj : a_{kj} \ne 0 \mbox{ and } k
\ne j \}$.
Observe that $G=\G(A_G)=\G(L_G)$, where $A_G$ and $L_G=D_G-A_G$ denote the adjacency matrix of $G$ and the Laplacian matrix of $G$, respectively (here $D_G$ is the diagonal  matrix of degrees).
Zero forcing has been studied in detail in the context of linear algebra. This is because the size of the minimum zero forcing set of a given graph $G$, which is called the \emph{zero forcing number} $\Z(G)$, is an upper bound to the maximum nullity (or maximum co-rank) over any field of $G$ \cite{AIM}; the maximum nullity is taken over all symmetric matrices $A$ such that $\G(A)=G$ (see  \cite{FH07} for background on the problem of determining maximum nullity).

Zero forcing appears then to be a valuable concept  in the study of graph-theoretic properties that are captured by generalized adjacency matrices. Indeed, there are important classical parameters introduced with this purpose, {e.g.}, the Colin de Verdi\'{e}re number, the Haemers bound, {etc.} It has to be remarked that questions about the maximum nullity of a graph are generally difficult problems and the zero forcing number does not constitute an exception: it was shown in \cite{AA08} that there is no poly-logarithmic approximation algorithm for the zero forcing number.

\subsection{Contribution of the paper and physical motivation}

 In this paper, we consider the dynamics of a system defined on a
 network and relate the above notions and criteria of
 controllability with the graph theoretic concept of zero forcing. Abstractly,
 we consider a graph $G$ and a subset $S=\{j_1,\dots,j_s\}$ of its vertices $V(G)=\{1,\dots,n\}$. The
 dynamics are that  of a quantum system (\ref{scromat}) where the
 Hamiltonian is allowed to take the values $\{ A, \be_{j_1} \be_{j_1}^T,
 \ldots, \be_{j_s} \be_{j_s}^T\}$.  Here $A$ is the adjacency matrix $A_G$ of $G$, Laplacian matrix $L_G$ of $G$, or more generally a real symmetric matrix such that $\G(A)=G$ with all nonzero off-diagonal entries of $A$ having the same sign  (which is the
typical situation in  transport models). The vectors $\{\be_{j_1}, \ldots, \be_{j_s}\}$ are the
 characteristic vectors\footnote{The vector $\be_j$
 has the $j$th entry equal to one and every other entry equal to zero and is also called the $j$th {\em standard basis vector}.} of the vertices in $S$. In this way, we can associate a linear
 system (\ref{sistema12}) with $A$ and $\bb_1=\be_{j_1},\ldots, \bb_s=\be_{j_s}$. The main result of the present
 paper says that  controllability in the quantum sense, expressed by the
 Lie algebra rank condition, and   controllability in the sense of
 linear systems, expressed by the Kalman rank condition, are
 equivalent conditions. Moreover, if
 the set $S$ (corresponding to $\be_{j_1}, \ldots, \be_{j_s}$) is a zero
 forcing set, then these  equivalent controllability conditions are true (the converse is false). The first of these results is along the same lines as the main result of \cite{GS10} which considers the case of quantum dynamics
 switching between the Hamiltonian $A$ %(weighted adjacency  or Laplacian matrix)
 and
 $\bz \bz^T$,  where $\bz=\sum_{j\in S} \be_j$, and establishes the
 connection between controllability (quantum and linear). % with the zero forcing property, in the spirit  of controllability of quantum systems via graph  theoretic properties of the interaction between the parties.
 As  mentioned above, these characterizations avoid lengthy calculations of the Lie algebra generated by a given set of Hamiltonians and replace them with more easily verified graph theoretic and linear algebra tests.

On  physical grounds,  our motivation for considering a Hamiltonian specified by a matrix with the given graph   comes from the study of {\it continuous time quantum walks} which model transport phenomena in many physical and  biological systems \cite{Caruso11}. A recent review is given in \cite{Blumen}. Most of the studies  consider this sole Hamiltonian and concern the statistical (diffusion) properties of the dynamics. We add here the Hamiltonians $\be_j \be_j^T$ where $\be_j$ is the characteristic vector of a given node of the network  and study the nature of the states that the resulting dynamics can achieve, in particular whether an arbitrary (unitary) state transfer can be achieved between the states of the quantum system. The Hamiltonians $\be_j\be_j^T$ model a prescribed energy difference between the corresponding node and all the other nodes of the network which are assumed to be at the same energy level. Therefore the dynamics
is the alternating of a diffusion process (modeled by the Hamiltonian $A$) and a rearrangement of the energies of the various states by selecting one of the states as high energy state and all the other at the same (lower) energy.
%$\bz_j \bz^T_j:=\be_1\be_1^T$.
%Leaving on the side the connection with quantum mechanics, the notion of a \emph{controllable pair} from %\cite{GS10, God10} was originally defined with the purpose of introducing an idea of controllability at %the algebraic graph theory level, linking algebraic and dynamical properties to the structure of the %graph. One of our goals is to clarify the interplay between zero forcing and the Kalman rank condition, %and therefore bridging between a standard notion in control theory and a type of diffusion process on %graphs. This framework has the potential to be an extra useful addition to the toolbox for the analysis of %complex networks.

Theoretical research in network theory has focused on a number of discrete time, deterministic diffusion processes on graphs. While zero forcing has not been studied in this context, there are two directions of research that are closely related: as it was already noted in \cite{AA08}, the threshold model introduced for studying influence in social networks shares with zero forcing certain issues underlying its computational complexity \cite{KKT03}; the model of complex networks controllability recently proposed in \cite{LSB11} also makes a natural use of the Kalman rank condition and it singles out certain combinatorial properties to determine when the condition is satisfied. Determining whether zero forcing has a place in the metrology of complex networks is a point worth further interest.

The paper is organized as follows. In Section \ref{sprelim} we introduce notation and give
background and basic results  concerning Lie algebras that will
be used in the following sections.   The connection between quantum (Lie
algebraic) controllability and the Kalman criterion for linear systems is
established in Section \ref{scont}.  There we also prove  the  converse of the main result of \cite{GS10}.  The relation with the zero
forcing property is established in Section \ref{Zerforc}, and  Section \ref{send} contains concluding remarks.

\section{Lie algebra terminology and preliminary results}\label{sprelim}

Standard material on Lie algebras can be found in \cite{Humph}.  For $A_1,\dots,A_k\in\Cnn$,  $\lie{A_1,\dots,A_k}$ denotes the real Lie algebra generated by $A_1,\dots,A_k$ under addition, real scalar multiplication, and the commutator operation.  Let $\Hn(\R)$ denote the real vector space of symmetric matrices.  For $A\in\Hn(\R)$, the notation  $A=[a_{kj}]$ means for $k<j$ the $(k,j)$ and $(j,k)$ entries of $A$ are both $a_{kj}$.   Observe that $A=[a_{kj}]\in\Hn(\R)$ can be expressed as
\[A=\sum_{k=1}^na_{kk}\be_k{\be_k}^T+\sum_{k<j}a_{kj}(\be_k{\be_j}^T+\be_j{\be_k}^T).  \]
The following proposition is well known (a proof appears in \cite{GS10}).  It provides a link between an appropriate
 Lie algebra of real matrices and the Lie algebra rank condition of quantum control theory, thereby allowing us to work with real matrices only.  Recall that the Lie algebra consisting of all real $n\x n$ matrices is denoted by $gl(n,\R)$,  $sl(n,\R)$ denotes the Lie algebra of real $n\times n$  matrices with zero trace, $u(n)$ denotes the Lie algebra of all skew-Hermitian (complex) $n\times n$ matrices, and $su(n)$ denotes the Lie algebra of all skew-Hermitian (complex) $n\times n$ matrices with zero trace.  All these Lie algebras are considered as vector spaces over the field of real numbers.
 \begin{prop}\label{cplxeq} For $A_1,\dots,A_k\in\Hn(\R)$, \[\lie{A_1,\dots,A_k}=gl(n,\R) \iff \lie{iA_1,\dots,iA_k}=u(n). \]
\end{prop}

The next lemma is used in the proof of Theorem \ref{mtxthm} in the next section.

\begin{lem}\label{sln}  Let $A,B_1,\dots,B_s\in\Hn(\R)$, with $s \geq 1$.  Define $\LL:=\lie{A,B_1,\dots,B_s}$ and let
 $\hat\LL$  denote the smallest ideal of $\LL$ that contains $B_i, i=1,\dots, s$.  If
 $\LL=gl(n,\R)$ and $\tr B_k\ne 0$ for some $B_k$, then $\hat\LL=gl(n,\R)$.
\end{lem}

\bpf    For $n=1$ the result is clear, so assume $n\ge 2$, $\LL=gl(n,\R)$, and $\tr B_k\ne 0$ for some $B_k$. Observe that $\LL:=\lie{A,B_1,\dots,B_s}$ is spanned by $A$ and $\hat\LL$. %, i.e.,  $\lie{A,B_1,\dots,B_s}= \Span(A) \, + \, \hat \LL$.
Since $\LL=gl(n,\R)$,
 %$\lie{A,B_1,\dots,B_s}=gl(n,\R)$,
we have
  \[
  [gl(n,\R),gl(n,\R)]=[\Span(A)\, + \, \hat \LL \, , \, \Span(A)\, + \, \hat \LL ] \subseteq \hat\LL.
  \]
  %$\hat\LL$ is a Lie ideal of $\Rnn$.  Furthermore,  Thus
  It is known that $[gl(n,\R), gl(n,\R)]=sl(n,\R)$, because $[gl(n, \R), gl(n,\R)]$  is a nonzero ideal in $sl(n,\R)$ and  $sl(n,\R)$ is a simple Lie algebra.  Since $\dim sl(n,\R)=n^2-1$ and $B_k\not\in sl(n,\R)$, $\dim \hat\LL \ge n^2$. Thus $\hat\LL=gl(n,\R)$.  \epf

The next lemma is used in the proof of Theorem \ref{MainT} in the next section.  Let $\LL$ be a Lie algebra, $A\in \LL$, and let $\KK$ be a subspace of $\LL$.  Recall that the
operation $ad_A$ is defined as $ad_A(B):=[A,B]$, and the {\em normalizer}   of $\KK$ is
\[N_{\LL}(\KK)=\{A\, : \, [A,B]\in\KK \mbox{ for all } B\in\KK\}.\]
It follows from the Jacobi identity that $N_{\LL}(\KK)$ is a subalgebra of $\KK$ \cite[p. 7]{Humph}.

\begin{lem}\label{Lemmaj} Let $A, L\in\Hn(\R)$.  Assume $\lie{iA, iL}=u(n)$ and define
 \begin{equation}\label{defin}
\sym:=\Span(\{ad_{iA}^{k_1}\, ad_{iL}^{k_2} \, \cdots \, ad_{iA}^{k_{s-1}}\,
ad_{iL}^{k_{s}}\, [iA,iL]\}), \end{equation}
 where $s$ and $k_1,\ldots,k_s$
are nonnegative integers. Then $\sym=su(n)$.
\end{lem}

\bpf   First note that $[iA,iL]\ne 0$  because we have assumed
that $iA$ and $iL$ generate $u(n)$. Clearly $iA, iL\in N_{u(n)}(\sym)$.  Since $N_{u(n)}(\sym)$ is a subalgebra of $u(n)$ and $iA$ and $iL$ generate $u(n)$, $N_{u(n)}(\sym)=u(n)$.  Thus $\sym$ is an ideal of $u(n)$.  Notice that $\sym \subseteq su(n)$ since $[iA,iL]$
is skew-Hermitian  with zero trace and $iA$ and $iL$ are
skew-Hermitian. Since $\sym$ is an ideal of $u(n)$, $\sym$ is an ideal of $su(n)$, and $\sym\ne \{ 0 \}$.  Since $su(n)$ is a simple Lie algebra, by
definition it has only the trivial ideals $\{ 0 \}$ and $su(n)$.
Therefore ${\cal S}=su(n)$. \epf

For $A\in\Hn(\R)$ and $Z=\{\bz_1,\dots,\bz_s\} \subset \Rn$, the {\em real Lie algebra generated by $A$ and $Z$} is defined as
\begin{equation}\label{E1}
\LL(A,Z):=\lie{A,\bz_1{\bz_1}^T,\dots,\bz_s{\bz_s}^T}.
 \end{equation}

%%%%%%%%%%%%%%%%%%%%%%%%%%%%%%%%%%%%%%%%%%%%%%%%%%%%%%%%%
%%%%%%%%%%%%%%%%%%%%%%%%%%%%%%%%%%%%%%%%%%%%%%%%%%%%%%%%%
%%%%%%%%%%%%%%%%%%%%%%%%%%%%%%%%%%%%%%%%%%%%%%%%%%%%%%%%%

\section{Controllability and walk matrices}\label{scont}

For $A\in\Hn(\R)$ and $Z=\{\bz_1,\dots,\bz_s\} \subset \Rn$, the  {\em extended walk matrix} of $A$ and $Z$ is the $n\x (ns)$ real matrix
\begin{equation}\label{E2}
\W(A,Z):=[\bz_1, A\bz_1,\dots,A^{n-1}\bz_1,\dots,\bz_s, A\bz_s,\dots,A^{n-1}\bz_s].
\end{equation}
A special case is when $Z=Z_S:=\{ \be_j: j \in S\}$ (with $\be_j$ denoting the $j$-th standard basis vector) for some subset $ S \subseteq V(G) $ for a graph $G$ and $A$ is the adjacency matrix $A_G$ of the graph. In this case, the relevant walk matrix is $\W(A_G, Z_S)$.

For $s=1$ the connection between the walk matrix $\W(A, Z)$ in (\ref{E2}) and the Lie algebra $\LL(A,Z)$ in (\ref{E1}) was studied in \cite{GS10}. It was shown \cite[Lemma 1]{GS10} that $\rank \W(A,\{\bz\})=n$ implies $\LL(A,\{\bz\})=gl(n,\R)$, or equivalently,  $\lie{iA,i\bz{\bz}^T}=u(n)$ (cf. Proposition \ref{cplxeq}). The next theorem states that the  converse of this result is also true.

\begin{thm}\label{MainT} Consider a matrix
$A$ in $\Hn(\R)$ and a vector $\bz \in \Rnn$. Then,
$\lie{iA,i\bz{\bz}^T}=u(n)$ (or equivalently
$\LL(A,\{\bz\})=gl(n,\R)$) implies that  $\rank \W(A,\{\bz\})=n$.
\end{thm}

\begin{proof}
 The equivalence of the hypotheses is justified by Proposition \ref{cplxeq}.  The result is clear if $n=1$, so assume $n\ge 2$.
We
use a contradiction argument. Assume the rank of the walk matrix $\W(A,\{\bz\})$
 is  less than $n$ but $\lie{iA, iL}=u(n)$, where $L:=\bz\bz^T$. There exists a vector
 $\bx \in \Cn$ such that $\bx^*\W(A,\{\bz\})=0$.
Consider the rank $1$ matrix $D:=\bx\bx^*$. We claim that $D$ commutes with every
matrix in ${\cal S}$, where $\sym$ is as in (\ref{defin}). To see this, notice
that from (\ref{defin}), all elements in ${\cal S}$ are linear
combinations of monomials of the form $M=A^{k_1}L^{k_2}A^{k_3}\cdots
L^{k_{p-1}}A^{k_p}$, for some $p\geq 1$, $k_j \geq 0$,  and $L$
appearing at least once with exponent  greater  than zero.
When multiplying $D$ with $M$, with $D$ on the left,   write $M$ as
$A^{k_1}LY$ for some matrix $Y$, so we have \begin{equation}\label{zero1} DM=
DA^{k_1}LY =\bx\bx^*A^{k_1}\bz\bz^*Y=0, \end{equation} which follows
immediately from the condition $\bx^*\W(A,\{\bz\})=0$ for $n-1 \geq k_1\geq 0$, and by using
the Cayley-Hamilton theorem for $k_1 \geq n$. Analogously, when
multiplying $D$ on the right of $M$, we write $M$ as $QLA^{k_p}$,
for some matrix $Q$, and we have \begin{equation}\label{iop} MD= QLA^{k_p}D=
Q\bz\bz^*A^{k_p}\bx\bx^*= 0,  \end{equation} since $\bx^*A^{k_p}\bz=0$ also implies
$\bz^*A^{k_p}\bx=0$. Therefore $D$ commutes with all elements of ${\cal S}$.

Observe  that since $su(n)$ is simple, $su(n)$ is an
irreducible representation of $su(n)$.  Therefore, since  $D$
commutes with  all elements of ${\cal S}$, it follows from Schur's Lemma
that $D$ must be a scalar multiple of the identity \cite[p. 26]{Humph}.
 However this is not
possible since $D$ has rank $1$.  This gives the desired
contradiction and thus completes the proof.
\end{proof}

We study the generalization of this result to multiple vectors ($s \geq 1$)  but for matrices $A$ and vectors $\bz_1,\ldots,\bz_s$ related to a connected graph $G$. In particular, $\G(A)=G$,  all nonzero off-diagonal entries of $A$ have the same sign, and $\be_{j_1},\ldots,\be_{j_s}$ will be the characteristic vectors associated to a subset $S$ of the vertices.  In the next section we will   relate this to the zero forcing property of the set $S$. In the context of graphs, it is important to consider multiple vectors because if $G$ is a graph and $\rank A_G\le \ord G-2$, then $\rank \W(A_G,\{\bz\})<n$ for any one vector $\bz$. On the other hand we will see that if  $S$ is a zero forcing set for $G$ and $\G(A)=G$, then $\LL(A,\{\be_j:j\in S\})=\Hn(\R)$ (see Theorem \ref{DBzf} below).

The next definition extends the definition given in \cite{God10}  (and implicitly in \cite{GS10}) of an associative algebra  that links the walk matrix and controllability.

\begin{defn} {\rm For $A\in\Hn(\R)$ and %orthogonal set
$Z=\{\bz_1,\dots,\bz_s\} \subset \Rn$, define \[P(A,Z):=\{A^m\bz_k{\bz_j}^TA^\ell :  1\le k,j\le s, 0\le m,\ell\le n-1\}.\] }\
\end{defn}
%The next theorem extends Theorem 2.2 in \cite{God10} and Lemma 1 in \cite{GS10}.

\begin{rem}{\rm  For $A\in\Hn(\R)$ and
$Z=\{\bz_1,\dots,\bz_s\} \subset \Rn$,  the associative algebra generated by $P(A,Z)$ is equal to  $\Span P(A,Z)$, because
%\beq\label{1Deq}
\[(A^m\bz_k{\bz_j}^TA^\ell)(A^g\bz_p{\bz_q}^TA^h)=({\bz_j}^TA^{\ell+g}\bz_p)A^m\bz_k{\bz_q}^TA^h\mbox{ and } {\bz_j}^TA^{\ell+g}\bz_p\in\R. \]}%\eeq
\end{rem}

\begin{lem}\label{prod}  For $A\in\Hn(\R)$ and
$Z=\{\bz_1,\dots,\bz_s\} \subset \Rn$,
 $\rank \W(A,Z)=n$ if and only if $\Span P(A,Z)=\Rnn$. \footnote{As a vector space, $\Rnn$ is the same as $gl(n,\R)$. We use the latter notation when we want to stress the Lie algebra structure on $gl(n,\R)$.}
\end{lem}
\bpf  Clearly $\rank \W(A,Z)=n$ if and only if $\range \W(A,Z)=\Rn$.
First assume $\rank \W(A,Z)=n$. For any matrix $M\in\Rnn$ with $\rank M=r$, there exist vectors $\bx^{(q)},\by^{(q)}$, $q=1,\dots r$, such that $M=\sum_{q=1}^r \bx^{(q)}{\by^{(q)}}^T$.  Since $\range \W(A,Z)=\Rn$, each $\bx^{(q)}$ is expressible as a linear combination of the columns of $\W(A,Z)$, {i.e.}, as a linear combination of vectors of the form $A^m\bz_k$, and similarly for $\by^{(q)}$.  Thus each $\bx^{(q)}{\by^{(q)}}^T$, and hence $M$,  is expressible as a linear combination of $A^m\bz_k{\bz_j}^TA^\ell$. Thus the matrices of the form $A^m\bz_k{\bz_j}^TA^\ell$ span $\Rnn$.

  For the converse, observe that if $B=\{\bb_1,\dots,\bb_r\}$ is a basis for $\range \W(A,Z)$, then \[\Span P(A,Z)=\Span(\{\bb_k{\bb_j}^T: 1\le k,j\le r\}).\]  If $n>r=\rank \W(A,Z)$, then $\dim \Span P(A,Z)\le r^2<n^2=\dim\Rnn$, so the matrices in $P(A,Z)$ cannot span $\Rnn$.
\epf

The {\em distance} between two distinct vertices $u$ and $v$ of a connected graph $G$, denoted by $d(u,v)$ is the minimum number of edges in a path from $u$ to $v$.

\begin{lem}\label{connect}   Let $A\in\Hn(\R)$  such that   $\G(A)$ is connected and  all nonzero off-diagonal entries of $A$ have the same sign. %\footnote{i.e., $A$ is a weighted adjacency or Laplacian matrix.}
If $k,j\in\{1,\dots,n\}$ and $k\ne j$, then $(A^{d(k,j)})_{kj}\ne 0$.
\end{lem}

\bpf  Let $d:=d(k,j)$.  The entry $(A^{d})_{kj}$ is a  sum of terms which are each the  product of $d$ nonzero entries of $A$.  Since $d$ is the distance between $k$ and $j$, only off-diagonal entries can appear in this product.  Thus every term has the same sign and $(A^d)_{kj}\ne 0$.
\epf

\begin{lem}\label{lieprod}   Let $A\in\Hn(\R)$  be such that   $\G(A)$ is connected and all nonzero off-diagonal entries of $A$ have the same sign.  Let
$S\subseteq\{1,\dots,n\}$ and
$Z=\{\be_j:j\in S\}$ be the subset of standard basis vectors. %, each one corresponding to a node(in $S$) of the graph $G$.
Then
 $\Span P(A,Z)\subseteq\LL(A,Z)$.
\end{lem}
\bpf %We show  that $A^m\bz_k\bz_j A^\ell\in \LL(A,Z)$ for all $j,k\in\{1,\dots, s\}, \ell,m\in\{0,\dots, n-1\}$.
 The proof of Lemma 1 in \cite{GS10} shows that for any real symmetric matrix $A$ and vector $\bz$,  $A^m\bz{\bz}^TA^\ell\in \LL(A,\{\bz\})$ for all $ m,\ell\in\{0,\dots, n-1\}$. Applying this, we obtain that   $A^m\be_j{\be_j}^TA^\ell\in \LL(A,Z)$  for all $j\in\{1,\dots, s\}, m,\ell\in\{0,\dots, n-1\}$. The result will follow if we are able to show that  $A^m \be_k {\be_j}^T A^\ell\in \LL(A,Z)$  for all $k,j\in\{1,\dots, s\}, m,\ell\in\{0,\dots, n-1\}$, with $k$  different from $j$.

 Consider the  distance $d(k,j)$ between the nodes $k$ and $j$ in $\G(A)$, which is  $\leq n-1$ because $\G(A)$ is
 connected. From the fact that both $\be_k \be^T_k$ and $A^{d(k,j)} \be_j \be_j^T$ are in $\LL(A,Z)$, we have  in $\LL(A,Z)$,
 \[[\be_k {\be_k}^T, A^{d(k,j)} \be_j{\be_j}^T]=\be_k{\be_k}^TA^{d(k,j)} \be_j{\be_j}^T-A^{d(k,j)}\be_j{\be_j}^T\be_k{\be_k}^T=({\be_k}^TA^{d(k,j)} \be_j)
\be_k{\be_j}^T.\]
%because $Z$ is an orthonormal set.
It follows from Lemma \ref{connect} that
${\be_k}^TA^{d(k,j)}{\be_j} \ne 0$, and so $\be_k{\be_j}^T\in\LL(A,Z)$.

Then
\begin{eqnarray}
 [A^m\be_k{\be_k}^T,\be_k{\be_j}^T] &=& A^m\be_k{\be_k}^T\be_k{\be_j}^T-\be_k{\be_j}^TA^{m}\be_k{\be_k}^T      \nonumber \\
   &=& A^m\be_k{\be_j}^T-({\be_j}^TA^{m}\be_k)\be_k{\be_k}^T.\ \nonumber
\end{eqnarray}
So, $A^m\be_k{\be_j}^T\in\LL(A,Z)$.  Similarly,  $\be_k{\be_j}^TA^\ell\in\LL(A,Z)$.
Finally,
\begin{eqnarray}
 [A^m\be_k{\be_k}^T,\be_k{\be_j}^TA^\ell] &=& A^m\be_k{\be_k}^T\be_k{\be_j}^TA^\ell-\be_k{\be_j}^TA^{m+\ell}\be_k{\be_k}^T      \nonumber \\
   &=& A^m\be_k{\be_j}^TA^\ell-({\be_j}^TA^{m+\ell}\be_k)\be_k{\be_k}^T.\ \nonumber
\end{eqnarray}
So, $A^m\be_k{\be_j}^TA^\ell\in\LL(A,Z)$.
\epf

The following theorem establishes the connection between quantum Lie algebraic controllability and the rank condition for an extended walk matrix modeled on a graph.

\begin{thm} \label{mtxthm} Let $A\in\Hn(\R)$  such that   $\G(A)$ is connected and all the nonzero off-diagonal elements of $A$ have the same sign.  Let
$S\subseteq\{1,\dots,n\}$ and
$Z=\{\be_j:j\in S\}$ be a subset of standard basis vectors. %  (each one corresponding to a node of the graph $G$ in $S$).
Then
 $\rank \W(A,Z)=n$
if and only if $\LL(A,Z)=gl(n,\R)$.
\end{thm}
\bpf  %With our previous notation, $P(A,Z)=\{{A_G}^k\be_i{\be_j}^T{A_G}^\ell : 0\le k,\ell\le n-1, i,j\in S\}$.
By Lemma \ref{prod}, $\Span P(A,Z)=\Rnn$ if
and only if $\rank \W(A, Z)=n$, so it suffices to show that  $\Span P(A,Z)=\Rnn$
if and only if $\LL(A,Z)=gl(n,\R)$.
By Lemma \ref{lieprod}, $\Span P(A,Z)\subseteq\LL(A,Z)$, so  $\Span P(A,Z)=\Rnn$ implies $\LL(A,Z)=gl(n,\R)$.   %Since $\bz_i{\bz_i}^T\in $\bz_i{\bz_i}^TA^\ell\in \LL(A,Z)$ for all $i\in\{1,\dots, s\}, k,\ell\in\{0,\dots, n-1\}$
For the converse, assume $\LL(A,Z)=gl(n,\R)$.  Then,
by Lemma \ref{sln}, $\hat\LL=gl(n,\R)$, where   $\hat\LL$ is the smallest ideal of $\LL(A,Z)$ that contains $\be_j{\be_j}^T,j=1,\dots,s$.   It is clear that $\hat\LL\subseteq\Span P(A,Z)$, so  $\Span P(A,Z)=\Rnn$.
\epf

%Our main result of this section follows from   Theorem \ref{mtxthm} and Proposition \ref{cplxeq}.
\begin{cor} \label{mtxcor38} Let $A\in\Hn(\R)$  such that   $\G(A)$ is connected and all the nonzero off-diagonal elements of $A$ have the same sign, and let $S\subseteq \{1,\dots,n\}$.  Then
 $\rank \W(A,\{\be_j:j\in S\})=n$
if and only if $\lie{iA,\{i\be_j{\be_j}^T:j\in S\}}=u(n)$, i.e., the quantum system associated with the Hamiltonians $iA$ and $i\be_j \be_j^T$, $j=1,\ldots,s$, is controllable.
\end{cor}

Observe that for any connected graph $G$, the adjacency matrix $A_G$ and the Laplacian matrix $L_G$ satisfy the hypotheses of Theorem \ref{mtxthm} and Corollary \ref{mtxcor38}.

 The result of \cite{GS10} for the case $s=1$ showing that $\rank \W(A, \{ \bz \})=n$ implies $\lie{iA,i\bz{\bz}^T}=u(n)$, (and the converse proved in Theorem \ref{MainT} in this paper)  were proved in reference to systems on graphs. The proofs however go through for an arbitrary symmetric   matrix $A$ and vector $\bz$. It is natural to ask whether the conditions on  the  matrix $A$ that we have used in Theorem \ref{mtxthm} are really necessary.
To this purpose, we can observe that the result is not true if we give up either of the hypotheses that 1) $\G(A)$ is connected or 2) the off-diagonal entries of  $A$ have the same sign, as shown in the next two examples.

\begin{ex}{\rm
To see the necessity of assuming that $\G(A)$ is connected, consider a   block diagonal matrix $A=\mtx{A_1 & 0 \\ 0 & A_2}$ with $A_1$ and $A_2$ symmetric matrices of dimensions $n_1$ and $n_2$, respectively, with $n_1+n_2=n$, and $\bz_1$ and $\bz_2$ two vectors that have zeros in the last $n_2$ or first $n_1$ entries, respectively, and such that the corresponding matrices $\W(A_1, \{\bz_1\})$ and   $\W(A_2, \{\bz_2\})$ have ranks $n_1$ and $n_2$, respectively. In this case, the walk matrix $\W(A, \{ \bz_1, \bz_2 \})$ has rank $n$, but the Lie algebra generated by $A$, $\bz_1\bz_1^T$, and $ \bz_2 \bz_2^T$ contains only block diagonal matrices.}
\end{ex}

\begin{ex}{\rm
To see the necessity of assuming that all nonzero off-diagonal entries of  $A$ have the same sign, consider $A=\mtx{0 & 1 & 0 & 1 \\
 1 & 0 & -1 & 0 \\
 0 & -1 & 0 & 1 \\
 1 & 0 & 1 & 0}$, and $Z=\{\be_1,\be_3\}$.   It is straightforward to verify that the walk matrix $\W(A, \{ \be_1, \be_3 \})$ has rank $4$.  However, $\rank \LL(A,\{\be_1,\be_3\})\le 8$, as can be seen as follows. Let

 \bea \LL: & = \Span(&\mtx{
 1 & 0 & 0 & 0 \\
 0 & 0 & 0 & 0 \\
 0 & 0 & 0 & 0 \\
 0 & 0 & 0 & 0
},
\mtx{
 0 & 0 & 0 & 0 \\
 0 & 1 & 0 & 0 \\
 0 & 0 & 0 & 0 \\
 0 & 0 & 0 & 1
},
\mtx{
 0 & 0 & 0 & 0 \\
 0 & 0 & 0 & 0 \\
 0 & 0 & 1 & 0 \\
 0 & 0 & 0 & 0
},
\mtx{
 0 & 0 & 0 & 0 \\
 1 & 0 & 0 & 0 \\
 0 & 0 & 0 & 0 \\
 1 & 0 & 0 & 0
},\\
& &
\mtx{
 0 & 1 & 0 & 1 \\
 0 & 0 & 0 & 0 \\
 0 & 0 & 0 & 0 \\
 0 & 0 & 0 & 0
},
\mtx{
 0 & 0 & 0 & 0 \\
 0 & 0 & 0 & 0 \\
 0 & 1 & 0 & -1 \\
 0 & 0 & 0 & 0
},
\mtx{
 0 & 0 & 0 & 0 \\
 0 & 0 & 1 & 0 \\
 0 & 0 & 0 & 0 \\
 0 & 0 & -1 & 0
},
\mtx{
 0 & 0 & 0 & 0 \\
 0 & 0 & 0 & 1 \\
 0 & 0 & 0 & 0 \\
 0 & 1 & 0 & 0
}).
%\Span\left(\{\}\right)
\eea
Since $[B,C]\in\LL$ for all $B,C\in\LL$, $\LL$ is a Lie subalgebra of $gl(4,\R)$.  Clearly $\dim\LL\le 8$ and
$\LL(A,\{\be_1,\be_3\})\subseteq \LL$.
}\end{ex}

%%%%%%%%%%%%%%%%%%%%%%%%%%%%%%%%%%%
\section{Zero forcing and controllability}\label{Zerforc}
The {\em neighborhood} of $v\in V(G)$ is $N(v)=\{w\in V(G): w \mbox{ is adjacent to } v\}.$

\begin{thm}\label{DBzf}
Let $A\in\Hn(\R)$ such that $\G(A)$ is connected and all the nonzero off-diagonal entries of $A$ have the same sign. Let $V:=\{ 1,2,\ldots,n\}$ be the set of vertices for $\G(A)$, and $S\subseteq V$  be a zero forcing set of $\G(A)$.
Then
\[
\LL(A,\{\be_{j}{\be_j}^{T}: j\in S\})=gl(n,\R).
\]
\end{thm}

\begin{proof}
After a (possibly empty) sequence of forces, denote by $T$ the set of currently black vertices, and assume that for all $k\in T$,
$\be_{k}{\be_k}^{T}\in\LL:=\LL(A,\{\be_{j}{\be_j}^{T}: j\in S\})$.  The hypotheses of   Lemma \ref{lieprod} are satisfied for $Z=\{\be_j : j\in T\}$, so
for all $k,\ell\in T$, $\be_{k}{\be_\ell}^{T}\in\LL$.

If $T \ne V$, then  there is a vertex $u\in T$ that has a unique neighbor $w$ outside
$T$. For that $u$ we have \[[\be_{u}\be_{u}^{T},A]=\sum_{m\in N(u)}\tilde a_{um}(\mathbf{e}_{u}\mathbf{e}_{m}^{T}-\mathbf{e}_{m}\mathbf{e}_{u}^{T}), \]
where $\tilde a_{um}:=a_{um}$ if $u<m$ and $\tilde a_{um}:=a_{mu}$ if $m<u$.
For all $m\in N(u)$ such that $m\ne w$,
$m\in T$, so $\be_{u}{\be_m}^{T}, \be_{m}{\be_u}^{T}\in\LL$.   Thus
$\mathbf{e}_{u}\mathbf{e}_{w}^{T}-\mathbf{e}_{w}\mathbf{e}_{u}^{T}\in\LL$.
Since
\[[\be_u{\be_u}^T, \mathbf{e}_{u}{\mathbf{e}_{w}}^{T}-\mathbf{e}_{w}{\mathbf{e}_{u}}^{T}]=\mathbf{e}_{u}{\mathbf{e}_{w}}^{T}+\mathbf{e}_{w}{\mathbf{e}_{u}}^{T},\]
 $\mathbf{e}_{w}\mathbf{e}_{u}^{T}, \mathbf{e}_{w}\mathbf{e}_{u}^{T}\in\LL$.  Then
\[[\mathbf{e}_{w}\mathbf{e}_{u}^{T},\mathbf{e}_{u}\mathbf{e}_{w}^{T}]=\mathbf{e}_{w}\mathbf{e}_{w}^{T}-\mathbf{e}_{u}\mathbf{e}_{u}^{T}\]
so $\mathbf{e}_{w}\mathbf{e}_{w}^{T}\in\LL$.
Since $S$ is a zero forcing set, we obtain $\mathbf{e}_{\ell}{\mathbf{e}_{m}}^{T}\in\LL$
for all $\ell,m\in V=\{1,\dots,n\}$, and thus we conclude that $\LL=gl(n,\R).$ \end{proof}

Applying Proposition \ref{cplxeq} we obtain the next corollary.

\begin{cor}
If $G$ is a connected graph, $A\in\Hn(\R)$,  all the nonzero off-diagonal entries of $A$ have the same sign, and  $S\subseteq V$ is a zero forcing set of $G$, then $\lie{iA,\{i\be_j{\be_j}^T:j\in S\}}=u(n)$ and the corresponding quantum system is controllable.
\end{cor}

Note that the converse of Theorem \ref{DBzf} is false.

\begin{ex}{\rm Consider the path on four vertices $P_4$ with the vertices numbered in order.  The set $\{\be_2\}$ is not a zero forcing set for $P_4$.  However, \[\W(A_{P_4}, \{\be_2\})=\mtx{0 & 1 & 0 & 2 \\
 1 & 0 & 2 & 0 \\
 0 & 1 & 0 & 3 \\
 0 & 0 & 1 & 0} \mbox{ and }\rank  \W(A_{P_4},\{\be_2\})=4,\] so $\LL(A_{P_4},\{\be_2\})=gl(n,\R)$ by Theorem \ref{mtxthm}.
}\end{ex}

\section{Conclusion}\label{send}

Motivated by the control and dynamics of systems modeled  on networks both classical and quantum, we have established a connection between various tests of controllability and the notion of zero forcing in graph theory. Lie algebraic quantum controllability is necessary and sufficient for linear (Kalman-like) controllability of an associated system and both notions are implied by the zero forcing property of the associated set of vertices. Linear systems have a very well developed  theory \cite{Kailath} and it is an open question to investigate to what extent this analogy can be further used to discover properties of quantum systems and systems on networks.

\vspace{3mm}

{\bf Acknowledgements.} The authors would like to thank Mark Hunacek for  productive discussions concerning Lie algebras. This work has been
done while Daniel Burgarth was  supported by EPSRC grant EP/F043678/1. Domenico  D'Alessandro is  supported by NSF under  Grant No.
ECCS0824085. Simone Severini  is a Newton International Fellow.  Michael Young's postdoctoral fellowship is supported by NSF through DMS 0946431.

%%%%%%%%%%%%%%%%%%%%%%%%%%%%%%%%%%%%%%%%%

\end{document}